\documentclass[11pt,fleqn]{article}
\usepackage[utf8]{inputenc}
\usepackage{fullpage}

\usepackage{amsthm}
\usepackage[dvipsnames]{xcolor}
\usepackage{hyperref}
\usepackage{todonotes}
\usepackage{qcircuit}
\usepackage{cite}
\usepackage{tikz}
\usepackage{tikzsymbols}
\usepackage{titling}
\usepackage{textcomp}
\usepackage{mdframed}

\usepackage{thm-restate}

\usepackage{libertinust1math,amsmath,amssymb}
\usepackage{newtxtext}
\usepackage[T1]{fontenc}

\usepackage{tikz}
\usepackage{pifont}
\usepackage{bbding}

\usepackage{tikz}
\usetikzlibrary{shapes,backgrounds}

\usepackage{setspace}
\usepackage{nirkhe}

\allowdisplaybreaks 

\usepackage{tikz}
\usetikzlibrary{shapes,backgrounds}

\usepackage{cleveref}

\newtheorem{theorem}{Theorem}

    \crefname{lemma}{lemma}{lemmas}
\newtheorem{conjecture}{Conjecture}
    \crefname{conjecture}{conjecture}{conjectures}
\newtheorem{corollary}[theorem]{Corollary}
    \crefname{corollary}{corollary}{corollaries}

\newtheorem{fact}[theorem]{Fact}
    \crefname{fact}{fact}{facts}
\newtheorem{definition}{Definition}
    \crefname{definition}{definition}{definitions}

\newenvironment{xalign}[1][]{
    \subequations
    \label{#1}
    \align
}
{
    \endalign
    \endsubequations
}

\usepackage{authblk}

\makeatletter
\renewcommand{\maketitle}{\bgroup\setlength{\parindent}{0pt}
\begin{flushleft}
  \Large {\@title} \linebreak
  
  \normalsize \@author \linebreak
\end{flushleft}\egroup
}
\makeatother

\title{From hard problems in representation theory to hardness of cloning} %

\title{On the hardness of cloning and connections to representation theory}

\author[1]{Vojt\v{e}ch Havl\'{i}\v{c}ek}
\author[2]{Chinmay Nirkhe}

\affil[1]{\small IBM Quantum, \emph{Yorktown Heights, N.Y.}}
\affil[2]{\small University of Washington, \emph{Seattle, Wash.}}

\date{}

\begin{document}

\maketitle{}

\begin{abstract}
    The states accepted by a quantum circuit are known as the \emph{witnesses} for the quantum circuit's satisfiability. The assumption $\BQP \neq \QMA$ implies that no efficient algorithm exists for constructing a witness for a quantum circuit from the circuit's classical description. However, a similar complexity-theoretic lower bound on the computational hardness of cloning a witness is not known. 
    In this note, we derive a conjecture about cloning algorithms for \emph{maximally entangled states} over hidden subspaces which would imply that no efficient algorithm exists for cloning witnesses (assuming $\BQP \not \supseteq \NP$).
    The conjecture and result follow from connections between quantum computation and representation theory; specifically, the relationship between quantum state complexity and the complexity of computing Kronecker coefficients.
\end{abstract}

\section{The computational problem of cloning}

We say that a transformation from $n$-qubits to $2n$-qubits clones a quantum state $\ket{\psi} \in \qubits{n}$ if $\ket{\psi} \mapsto \ket{\psi} \otimes \ket{\psi}$. The no-cloning theorem proves that there is no single transformation for cloning a pair of non-orthogonal quantum states due to the linearity of quantum mechanics. For example, there is no linear transformation that simultaneously maps:
\begin{equation}
\ket{0} \mapsto \ket{0} \otimes \ket{0} \quad \text{and} \quad \ket{+} \mapsto \ket{+} \otimes \ket{+}, \text{ where: } \ket{+} \defeq \frac{\ket{0} + \ket{1}}{\sqrt{2}}.
\end{equation}

In this note, we focus on the computational hardness of cloning and do not consider information-theoretic issues with it. This is because we study the complexity of \emph{witness cloning}, the task of efficiently preparing the state $\ket{\psi} \otimes \ket{\psi}$ given as input a single copy of $\ket{\psi}$ and the classical description of a $\poly(n)$-sized quantum circuit $V$ that uniquely accepts $\ket{\psi}$. %

With access to the verifier $V$, the task is relieved of its information-theoretic barriers because a computationally unbounded algorithm exists for deriving the second copy of the state $\ket{\psi}$ from the description of $V$ without touching the initial copy. Therefore, the computational task of witness cloning is necessarily as easy as the task of constructing witnesses. The goal of this paper is to provide evidence that the task of cloning is not significantly easier than the aforementioned tasks -- i.e. the task of witness cloning does not have a polynomial-time quantum algorithm.

The definition can be generalized to verifiers $V$ which accept multiple orthogonal states. In general, we say a transformation clones the witnesses of $V$ if any witness of $V$ is mapped to a $2n$-qubit state such that the reduced density matrices on the first $n$ qubits and the last $n$ qubits would be individually accepted by $V$. This is because we regard the witness as a ``witness to the satisfiability of $V$''. In other words, if $\Ll$ is the subspace of witnesses of $V$, the transformation maps $\Ll$ to $\Ll \otimes \Ll$. It is important to note that for verifiers $V$ which accept multiple orthogonal states, the transformation may produce entangled states $\in \Ll \otimes \Ll$ due to linearity.
\medskip

\begin{center}\emph{How hard is the task of cloning witnesses, given access to the classical description of the verifier? Can we prove complexity-theoretic lower bounds for this task?} \end{center}

\subsection{The computational hardness of cloning}%

One technique for proving a lower-bound is showing that a \emph{black-box} oracle for witness cloning implies an unexpected complexity class collapse\footnote{Since the cloning task is not a decision problem, we cannot directly connect the complexity of this task to a decision complexity class.}, such as $\BQP = \QMA$. Equivalently, given access to a quantum oracle for cloning witnesses, there exists a polynomial-time quantum algorithm for deciding any $\QMA$ problem. We do not know such a statement and it is plausible that there exists an efficient witness cloning algorithm for all verification circuits. If an efficient witness cloning algorithm were to exist, it would imply that the complexity of constructing multiple copies of a witness is entirely concentrated in the complexity of constructing the first copy.

Despite not knowing complexity-theoretic evidence, we do know \emph{oracle} separations proving the computational hardness of witness cloning~\cite{aaronson-christiano} as well as hardness results based on cryptographic primitives~\cite{quantum-money-from-knots,broadbent2023uncloneablequantumadvice,zhandry-quantum-lightning,zhandry-group-actions,bostanci2024generalquantumdualityrepresentations}. Most of these are studied in the context of \emph{public-key quantum money} schemes, which can be thought of as an \emph{average-case} hardness analog of the worst-case hardness results we study in this work. We briefly survey these results.

\paragraph{Aaronson and Christiano's oracle separation} 
Aaronson and Christiano~\cite{aaronson-christiano} consider an oracle defined as the indicator function for $(\{0\} \times A) \sqcup (\{1\} \times A^\perp)$ where $A \subset \FF_2^n$ is a linear subspace of dimension $n/2$ and $A^\perp$ is the space of vectors orthogonal to it. Given access to such an oracle, there exists an efficient algorithm which only accepts the subset state $\ket{A}$ by measuring the input in the standard and Hadamard bases and comparing to the oracles for $A$ and $A^\perp$. Furthermore, they prove, using an adaptation of the adversary/hybrid method, that any algorithm cloning the state $\ket{A}$ for all such linear subspaces $A$ must make an exponential number of queries to the oracle, thereby proving a rigorous exponential lower bound on the task of cloning.
Zhandry~\cite{zhandry-quantum-lightning} shows that the oracle $(\{0\} \times A) \sqcup (\{1\} \times A^\perp)$ can be suitably hidden assuming \emph{quantum indistinguishability obfuscation} in such a way that one can prove hardness of the cloning task without the oracle but requires a purported and non-standard cryptographic assumption (that is yet to be broken).

\paragraph{Zhandry's quantum money}

In~\cite{zhandry-group-actions}, Zhandry constructs a quantum money scheme based on abelian group actions. He proved the security of the scheme in the \emph{generic group action} (GGA) model, assuming a new but natural strengthening of the discrete log assumption on group actions. Since the discrete logarithm problem has a query-efficient algorithm \cite{ettinger04}, the GGA model cannot be used to give unconditional hardness results. Zhandry's abelian group action scheme is an inspiration for the non-abelian group action results in this work. We also note a concurrent result by Bostanchi, Nehoran, and Zhandry~\cite{bostanci2024generalquantumdualityrepresentations} generalizing Zhandry's work to non-abelian group actions for constructions of quantum money, lightning, and fire, cryptographic notions quantifying hardness-of-cloning. We highly applaud this work and think that it is a big step towards understanding cryptographic hardness of these primitives. The specific goal of this work is to make progress on proving the hardness of cloning from a complexity-theoretic assumption such as $\BQP \not\supseteq \NP$.

\section{Our results}

We make progress towards a complexity lower bound for witness cloning. Specifically, we show that an efficient quantum algorithm for cloning witnesses implies $\BQP \supseteq \NP$, assuming a conjecture about cloning maximally entangled states over hidden subspaces. We believe that this conjecture (\Cref{conj:hidden-epr}) will prove to be significantly more tractable than the general problem.%

We identify a family of $\NP$-hard maximally entangled states over hidden subspaces, defined shortly. This is due to a nascent connection between quantum computation and the complexity of Kronecker coefficients. Our starting point is the observation of Bravyi~\emph{et. al.}~\cite{kronecker-coeff} that Kronecker coefficients are a $\mathsf{\#BQP}$ quantity. $\mathsf{\#BQP}$ is one possible quantum analog of $\SHARPP$ -- informally, a counting problem $\fn{f}{\bits^*}{\NN}$ is in $\mathsf{\#BQP}$ if there exists a polynomial time map $x \mapsto \Pi_x$ where $\Pi_x$ is an efficiently computable quantum projector such that $f(x) = \tr (\Pi_x)$. Kronecker coefficients, on the other hand, are positive integers derived from the symmetric group which describe irrep multiplicities in tensor products of irreps. Formally, for irreps $\rho^\mu, \rho^\nu, \rho^\lambda$ of the symmetric group $S_n$ (each irrep can be described by an $n$-box Young diagram), the Kronecker coefficients $\{m_{\mu \nu \lambda}\}$ are the non-negative integers satisfying the following equation:
\begin{equation}
    \rho^\mu \otimes \rho^\nu \equiv \bigoplus_\lambda \II_{m_{\mu \nu\lambda}}\otimes \rho_\lambda.
\end{equation}
Deciding if $m_{\mu \nu \lambda} > 0$ is $\NP$-hard by the result of Ikenmeyer, Mulmuley, and Walter~\cite{kronecker-np-hard}. Using this hardness result, for any $\NP$-problem, we construct an efficiently computable quantum projector $\Gamma_{\mu \nu \lambda}$ such that $\tr(\Gamma_{\mu \nu \lambda}) \in \{0, 1\}$ with $\tr(\Gamma_{\mu \nu \lambda}) = 1$ iff the $\NP$-problem is satisfiable. The unique state $\ket{\psi}$ such that $\ev{\Gamma_{\mu \nu \lambda}}{\psi} = 1$ is the maximally entangled state over a hidden subspace. For dimensions $d_1 \leq d_2$, and a subspace $\Pi \subseteq \CC^{d_2}$ of dimension $d_1$ and any basis $\ket*{b_1}, \ldots, \ket*{b_{d_1}}$ for $\Pi$, the corresponding \emph{maximally entangled witness state for the subspace} $\Pi$ is
\begin{equation}
        \ket*{\Phi_\Pi} \defeq \frac{1}{\sqrt{d_1}} \sum_{i = 1}^{d_1} \ket*{b_i} \otimes \ket*{b_i^*}.
        \label{eq:hidden-subspace}
\end{equation}
The state is independent of the choice of basis. If deciding if $\Pi \neq 0$ is $\NP$-hard, then we call the subspace \emph{hidden}. This is because there is no quantum polynomial-time algorithm for generating states in $\Pi$ assuming $\BQP \not \supseteq \NP$.

This produces our main result: if witness cloning is efficient for all verifiers, then the cloner must clone the maximally entangled states corresponding to a family of $\NP$-hard hidden subspaces. We believe that the maximally entangled nature of the witness prohibits the existence of an efficient cloner. We pose a conjecture~(\Cref{conj:hidden-epr}) which asserts that if an efficient cloner for the maximally entangled state over hidden subspace $\Pi$ exists, then there exists an efficient circuit generating a state in $\Pi$. Together, this conjecture and $\BQP \not \supseteq \NP$ imply that witness cloning is hard. 

We believe~\Cref{conj:hidden-epr} is a significant step towards lower bounding the complexity of witness cloning. Near the end of this document, we discuss the challenges of proving~\Cref{conj:hidden-epr} and why we believe it to be true.

\section{Definitions and problem statements}

\begin{definition}[Verification circuit]
A quantum circuit $V$ acting on $n$ qubits with $m$ ancilla is a verification circuit with completeness $c$ and soundness $s$ if the Hermitian operator
\begin{equation}
\H \defeq \qty(\II_2^{\otimes n} \otimes \bra{0}^{\otimes n}) \cdot V^\dagger \cdot \qty(\ketbra{1}_1 \otimes \II_2^{\otimes {n+m-1}} ) \cdot V \cdot \qty(\II_2^{\otimes n} \otimes \ket{0}^{\otimes n}),
\end{equation}
has no eigenvalues in the range $(s, c)$.
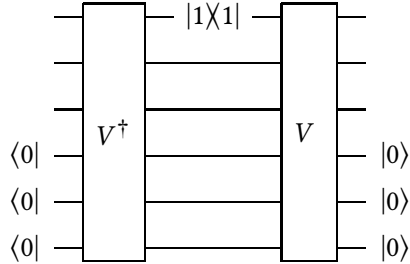
$$ \Qcircuit @C=1em @R=.7em { 
&\multigate{5}{V^\dagger} & \push{~\ketbra{1}~} \qw & \multigate{5}{V~} & \qw \\
&\ghost{V^\dagger} & \qw & \ghost{V~} & \qw\\
& \ghost{V^\dagger} & \qw & \ghost{V~} & \qw\\
\lstick{\bra{0}} & \ghost{V^\dagger} & \qw & \ghost{V~} & \rstick{\ket{0}} \qw \\
\lstick{\bra{0}} & \ghost{V^\dagger} & \qw & \ghost{V~} & \rstick{\ket{0}} \qw \\
\lstick{\bra{0}} & \ghost{V^\dagger} & \qw & \ghost{V~} & \rstick{\ket{0}} \qw
} $$
Equivalently, let $\Ll$ be the span of eigenvectors of eigenvalue $\geq c$ and $\Ll^\perp$ be the span of eigenvectors of eigenvalues $\leq s$. Then $\qubits{n} \cong \Ll \oplus \Ll^\perp$. We call $\Ll$ the space of accepting states by $V$ and $\Ll^\perp$ the space of rejecting states.
\label{def:verification-ckt}
\end{definition}

\begin{definition}[Cloning]
For a given verification circuit\footnote{In the specific case of the $\QMA$-complete problem of estimating the minimum eigenvalue of a local Hamiltonian instance, the cloning transformation definition can be interpreted as: 
\begin{definition}[Cloning of groundspaces]
Let $\H$ be a local Hamiltonian describing the evolution of a system of $n$-qubits. We say that a quantum transformation \emph{clones the groundspace of $\H$} if it maps any groundstate of $\H$ (an eigenvector of minimal eigenvalue) to a groundstate of the Hamiltonian $\H \otimes \II + \II \otimes \H$ which describes a $2n$-qubit system.
\end{definition} By standard arguments about Hamiltonian energy estimation, we note that any local Hamiltonian $0 \preccurlyeq \H \preccurlyeq \II $ can be transformed into a verification circuit with the property that $c - s \geq \gamma$, the spectral gap of the local Hamiltonian.} $V$ acting on $n$ qubits, we say that an $n$-qubit to $2n$-qubit transformation clones the accepting space of $V$ if it maps $\Ll$ to $\Ll \otimes \Ll$.
\end{definition}

As stated by Broadbent, Karvonen, and Lord~\cite{broadbent2023uncloneablequantumadvice}, the task of witness cloning can be viewed as the ``$1 \rightarrow 2$'' analog of the task of producing witness to quantum verification circuits which is the ``$0 \rightarrow 1$'' task.

\subsection{Computational task}

The computational task of cloning is to construct two accepting states from one accepting state $\ket{\psi}$ and the description of a verification circuit $V$. We will call this \emph{witness cloning} for $V$. We believe that there is no efficient transformation for cloning all verification circuits $V$ for completeness $2/3$ and soundness $1/3$.

\begin{conjecture}[The hardness of witness cloning]
There is no uniform quantum polynomial time algorithm which successfully clones the accepting space of every verification circuit $V$ of completeness $2/3$ and soundness $1/3$.
\end{conjecture}

\section{Representation theory notations and definitions}
Let $G$ be a finite group and let $\Ii_G$ be the set of distinct irreducible representations -- i.e. \textit{irreps} -- for $G$. Notationally, $\lambda \in \Ii_G$ will parameterize the irreps of $G$ with $\fn{\rho^\lambda}{G}{\CC^{d_\lambda \times d_\lambda}}$ being the corresponding irrep of dimension $d_\lambda$. 
Let $\chi^\rho(\cdot) = \tr(\rho(\cdot))$ be the character of a representation $\rho$. When $\rho = \rho^\lambda$ is an irrep, for brevity, we use the notation $\chi^\lambda$ to refer to $\chi^{\rho^\lambda}$. 
Two important (reducible) representations are the left- and right-regular representations, defined by their action on a $|G|$-dimensional vector space with orthonormal basis $\lbrace \ket{g} \rbrace_{g\in G}$:
\begin{equation}
    \rho_L(g) = \sum_{h \in H} \ketbra{gh}{h}, \qquad \rho_R(g) = \sum_{h \in H} \ketbra*{hg}{h}, \qquad \rho_C = \rho_L \rho_R = \rho_R \rho_L. \label{eq:lr-reg-reps}
\end{equation}
Every representation $\sigma$ of a finite group has a unique decomposition irreps that is defined by the multiplicity coefficients $m_{\sigma \lambda}$:
\begin{equation}
    \sigma \cong \bigoplus_{\lambda \in \Ii_G} \II_{m_{\sigma \lambda}} \otimes  \rho^\lambda.
\end{equation}
It is a standard fact that the left- and right-regular representations have decompositions with multiplicity coefficients $m_{\rho \lambda} = d_\lambda$. See Ref.~\cite{serre1977linear} for additional details. The unitary transformation that decomposes both left- and right regular-representations is the Fourier transform. For any finite group $G$, the Fourier transform unitary is given by
\begin{equation}
    \mathrm{FT} \defeq \sum_{\pi \in G} \sum_{\lambda \in \Ii_G} \sum_{i,j=1}^{d_\lambda} \sqrt{\frac{d_\lambda}{\abs{G}}} \rho^{\lambda}_{ij}(\pi) \ket{\lambda, i, j}\bra{\pi}. \label{eq:qft}
\end{equation}
Here $\ket{\lambda}$ is some classical representation of the irrep; in the case of the symmetric group, the irreps can be represented by Young tableaus.
For certain groups, including the symmetric group, application of the Fourier transform is efficient~ \cite{beals97, moore2003genericquantumfouriertransforms}. %

\paragraph{Representations of symmetric groups}
We describe all the transformations in this note with a finite group $G$, but our attention will eventually focus on $G = S_n$, the symmetric group --- i.e., the group of permutations on $n$ elements. For the symmetric group, the set of irreps, $\Ii_{S_n}$, can be parametrized by partitions $\lambda \vdash n$ or equivalent Young tableau of $n$ boxes. When the group is $S_n$ and the representation being studied $\sigma = \rho^\mu \otimes \rho^\nu$, then the coefficients are abbreviated as $m_{\mu \nu \lambda}$ and are known as Kronecker coefficients.

\paragraph{The computational complexity of Kronecker coefficients}
Bürgisser and Ikenmeyer~\cite{burgisserIkenmeyer2008} proved
that the Kronecker coefficient $m_{\mu \nu \lambda}$ is $\#\P$-hard to compute from input $(\mu, \lambda, \nu)$ for $\mu, \nu, \lambda \vdash n$. This choice of parametrization for the input size fixes the order of the underlying permutation group and is commonly referred to as \emph{unary} encoding. As mentioned previously, the problem of deciding if $m_{\mu \nu \lambda} > 0$ is $\NP$-hard by \cite{kronecker-np-hard}. See Ref.~\cite{panova2023computationalcomplexityalgebraiccombinatorics}, \cite{kronecker-coeff} and the references therein for additional discussion.

\section{A construction of hard verification algorithms from representation theory}

For every finite group $G$ and representation $\fn{\sigma}{G}{\CC^{D \times D}}$, there exists a \emph{positive operator valued measurement} (POVM) called the \emph{weak Fourier sampling}\footnote{Note that weak Fourier sampling is referred to as coarse Fourier sampling by~\cite{bostanci2024generalquantumdualityrepresentations}.} POVM related to measurements of a state in the Fourier basis.

\begin{restatable}{fact}{factwfs}(Weak Fourier Sampling)
\label{fact:weak-fourier-sampling}
Let $\sigma: G \rightarrow \CC^{D \times D}$ be a representation of a finite group $G$. For an irrep $\lambda \in \Ii_G$, let: 
\begin{equation}
    \Xi_\lambda = \Xi_\lambda^{(\sigma)} \defeq \frac{d_\lambda}{\abs{G}} \sum_{g \in G} \chi^\lambda(g)^* \sigma(g).
    \label{eq:wfs-eq}
\end{equation}
The set of operators $\lbrace \Xi_\lambda \rbrace_{\lambda \in \Ii_G}$ are orthogonal projectors and the corresponding measurement is called \emph{weak Fourier Sampling}.  This measurement can be implemented efficiently when both the conditional application of $\sigma$ and the quantum Fourier transform (\eqref{eq:qft}) for $G$ can be implemented efficiently.
\end{restatable}
\noindent The proof of this fact is given in~\Cref{sec:omitted-pfs}. 
It is convenient to work with the basis in which the representation $\fn{\sigma}{G}{\CC^{D \times D}}$ is block-diagonal:
\begin{equation}
\sigma = \bigoplus_{\lambda} 
\II_{m_{\sigma \lambda}} \otimes \rho^\lambda.
\label{eq:irrep-break-down}
\end{equation}
In this basis, $\Xi_\lambda$ can be seen as the projector onto the subspace corresponding to the $\lambda$-blocks; for notational simplicity, we also use $\Xi_\lambda$ as the name of the subspace it projects onto. It will be clear from context which is being used. In general, generating a state in $\Xi_\lambda$ is computationally hard due to known results on the hardness of Kronecker coefficients~\cite{kronecker-np-hard}.

\subsection{The complexity of generating states in $\Xi_\lambda$}

By considering $\Xi_\lambda$ in the basis given by~\eqref{eq:irrep-break-down}, we see that each block is $d_\lambda \times d_\lambda$ in size and the number of blocks is the multiplicity $m_{\sigma \lambda}$. Therefore, it follows that the dimensionality of the projector $\Xi_\lambda$ is precisely $m_{\sigma \lambda} \cdot d_\lambda$.

It follows that $\Xi_\lambda \neq 0$ if and only if $m_{\sigma \lambda} \neq 0$. In the particular case that $G = S_n$ and $\sigma = \rho^{\mu} \otimes \rho^{\nu}$, then, as previously stated, $m_{\sigma \lambda} = m_{\mu \nu \lambda}$, the Kronecker coefficient. %
Deciding if the Kronecker coefficient is positive is $\NP$-hard \cite{kronecker-np-hard} and this is equivalent to deciding whether $\Xi_\lambda = \Xi_\lambda^{(\mu \nu)} \neq 0$ or not. Since it is widely believed that $\BQP$ computations cannot decide $\NP$-hard problems,  there is no efficient quantum algorithm for \emph{producing} a state $\ket{\psi} \in \Xi_\lambda^{(\mu \nu)}$ from classical input $(\mu, \nu, \lambda)$. 

However, we can show that an efficient quantum algorithm exists for \emph{verifying} a state $\ket{\psi} \in \Xi_\lambda^{(\mu\nu)}$. This is due to~\Cref{fact:weak-fourier-sampling} combined with the fact that there exists an efficient classical implementation for $\sigma = \rho^{\mu} \otimes \rho^{\nu}$ for any two irreps $\mu, \nu \vdash n$.
The corresponding weak Fourier sampling POVM defined in Fact~\ref{fact:weak-fourier-sampling} has an efficient classical implementation for $\sigma = \rho^\mu \otimes \rho^\nu$:
\begin{restatable}{fact}{factefficientprojector}
There exists a polynomial-sized quantum circuit that implements the measurement $\lbrace \Xi_\lambda \rbrace_{\lambda \vdash n}$ from \Cref{fact:weak-fourier-sampling} for an $S_n$ representation $\sigma = \rho^\mu \otimes \rho^\nu$, where $\mu, \nu \vdash n$ label two irreps $\rho^\mu, \rho^\nu \in \Ii_{S_n}$.
\label{fact:efficientprojector}
\end{restatable}
\noindent A proof is provided in~\Cref{sec:omitted-pfs}.

\subsection{An improved construction}

The projector $\Xi_\lambda$ has dimension $m_{\sigma \lambda} d_\lambda$ because $\Xi_\lambda$ can be decomposed into $m_{\sigma \lambda}$ blocks each of dimensionality $d_\lambda$. Next, we come up with an additional test such that there is exactly one accepting state per block. We do this by adding a second test which forces a maximally entangled state between the original register and a second register of equal size. 

We first lift the representation $\fn{\sigma}{G}{\CC^{D \times D}}$ to a representation $\fn{\sigma'}{G}{\CC^{D^2 \times D^2}}$ defined by $\sigma \otimes \II_D$, a representation acting on two registers of equal size. The previously constructed notions of weak Fourier sampling still hold even when we apply the weak Fourier sampling to the first (left) of the two registers. The corresponding set of states passing the weak Fourier sampling of the left register will be states of the form:
\begin{equation}
    \ket*{\Psi_\lambda^{(\phi)}} \defeq \frac{1}{\sqrt{A_\lambda^{(\phi)}}} \sum_{h \in G} \chi^\lambda(h)^* \left( \sigma(h) \otimes \II_D \right) \ket{\phi} \label{eq:passing-improved-construction}
\end{equation}
where $A_\lambda^{(\phi)} \geq 0$ is a normalization constant. We will henceforth denote the set of states that pass the weak Fourier sampling test as $\Gamma_\lambda$.
In our improved construction, we would like to add an additional check which enforces that the states in~\eqref{eq:passing-improved-construction} require $\ket{\phi}$ equaling $\ket*{\Phi^+}$, the $D$-dimensional maximally entangled state between the two registers. Our construction cannot achieve this goal, but approaches this goal by forcing maximal entanglement between the two registers per $\lambda$-block. We define the following test to be included in addition to weak Fourier sampling. We call it ``internal state testing'' since it enforces a restriction on the ``internal'' states $\ket{\phi}$ allowed from~\eqref{eq:passing-improved-construction}.
\begin{mdframed}
\textbf{Internal state testing for $\ket{\psi} \in \CC^{D^2}$:}
\begin{enumerate}
    \item Generate the state $\mathrm{FT} \ket{0} \otimes \ket{\psi} = \displaystyle \frac{1}{\sqrt{\abs{G}}} \sum_{k \in G} \ket{k} \otimes \ket{\psi}$.
    \item Run the following 1-bit phase estimation circuit\footnote{It is a standard fact that the 1-bit phase estimation circuit accepts with probability $\half + \half \abs{\ev{U}{\psi}}^2$.} and accept if the measurement is 0. Here, 
    \begin{equation}\displaystyle U = \sum_{k \in G} \ketbra{k} \otimes \sigma(k) \otimes \sigma(k)^*.\end{equation}
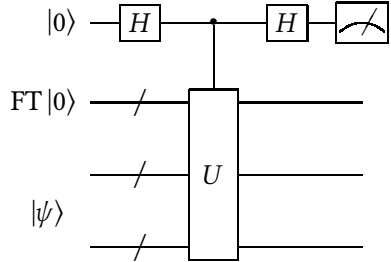
    $$
    \Qcircuit @C=1em @R=1.6em {
        & \lstick{\ket{0}} & \gate{H} & \ctrl{1} & \gate{H} & \meter \\
        & \lstick{\mathrm{FT} \ket{0}} & {/} \qw & \multigate{2}{U} & \qw & \qw \\
        & \lstick{} & {/} \qw & \ghost{U} & \qw & \qw \\
        & \lstick{} & {/} \qw & \ghost{U} & \qw & \qw 
            \inputgroup{3}{4}{1.35em}{\quad\ket{\psi}} \\
    }
    $$
\end{enumerate}
\end{mdframed}

This test exhibits a robust characterization. To describe the characterization, we first note that projector $\Xi_\lambda$ can be expressed as the direct sum of smaller projectors $\Xi_{\lambda,j}$ for $j = 1, \ldots, m_{\sigma \lambda}$ corresponding to the $\lambda$-blocks due to the decomposition expressed in~\eqref{eq:irrep-break-down}. 
Let $M_\lambda$ be the span of the maximally entangled states across each of the $\Xi_{\lambda,j}$ subspaces (as defined in~\eqref{eq:hidden-subspace}):
\begin{equation}
    M_\lambda \defeq \mathrm{span} \ \qty{ \ket*{\Phi_{\Xi_{\lambda,j}}} }_{j = 1, \ldots, m_{\sigma \lambda}}.
\end{equation}
The space $M_\lambda$ is a subspace of dimension $m_{\sigma \lambda}$ of $\Xi_\lambda \otimes \Xi_\lambda$ with one state per $\lambda$-block.

\begin{theorem}[Characterization of internal states]
If the internal state test passes with probability $1-\eps$ for a state in $\Gamma_\lambda$, then the internal state $\in \CC^{D \times D}$ is $\leq 2 \sqrt{2\eps}$-close to the subspace $M_\lambda$.
\label{thm:characterization-of-internal-test}
\end{theorem}

\begin{proof}
Within $\Gamma_\lambda$, the action of the representation $\sigma$ is $\II_{m_{\sigma \lambda}} \otimes \rho^\lambda$. Therefore, it suffices to consider the following lemma (\Cref{lem:interior-test-for-irrep-identity-prod}) which we prove independently in~\Cref{sec:omitted-pfs}.
\end{proof}

\begin{restatable}{lemma}{lemitfiip}
Consider $\sigma = \II_m \otimes \sigma_1$ where $\fn{\sigma_1}{G}{\CC^{D_1 \times D_1}}$ where $D = m D_1$ and $\sigma_1$ is an irrep. If the internal state testing procedure for representation $\sigma$ and state $\ket{\psi}$ passes with probability $1 - \eps$, then there a state $\ket{a} \in \CC^{m^2}$ such that
\begin{equation}
\norm{\ket{\psi} - \ket*{a} \otimes \ket*{\Phi^+}} \leq 2 \sqrt{2\eps}. \label{eq:interior-test-for-irrep-identity-prod}
\end{equation}
Here $\ket*{\Phi^+}$ is the maximally entangled state of $D_1$ dimensions: $\frac{1}{\sqrt{D_1}} \sum_{b = 0}^{< D_1} \ket{b} \otimes \ket{b}$.
\label{lem:interior-test-for-irrep-identity-prod}
\end{restatable}
\noindent We can combine both tests to completely describe the quantum verifier. We believe the following verification algorithm/circuit does not have an efficient witness cloner. In the next section, we detail the hardness results for this verification algorithm.
\begin{mdframed}
\textbf{Verification algorithm for the $(\sigma, \lambda)$ problem with input $\ket{\psi} \in \CC^{D^2}$:}
\begin{enumerate}
    \item Measure the POVM $\{\Xi_{\lambda'}\}_{\lambda' \in \Ii_G}$. Reject if the measurement $\lambda' \neq \lambda$.
    \item Perform internal state testing for $\sigma$ on the post-measurement state. Accept iff internal state testing accepts.
\end{enumerate}
\end{mdframed}

\begin{corollary}[Complete characterization]
Given a representation $\fn{\sigma}{G}{\CC^{D \times D}}$ and an irrep $\lambda \in \Ii_G$, if a state passes the $(\sigma,\lambda)$ verification test (above) with probability $1 - \eps$, then $\ket{\psi}$ is $\leq 3 \sqrt{2\eps}$-close to a unit state in $\Xi_\lambda$ of the form $\ket{a} \otimes \ket*{\Phi^+}$ where $\ket{a} \in \CC^{m_{\sigma \lambda}^2}$ when expressed in the basis from \eqref{eq:irrep-break-down}.
\label{cor:passing-all-three-tests}
\end{corollary}

\begin{proof}
By measuring the POVM and accepting only if the outcome is $\lambda$, we ensure that the state $\ket{\psi}$ is $\sqrt{\eps}$-close to the subspace $\Gamma_\lambda$. As internal state testing commutes with the POVM, the state must be $3\sqrt{\eps}$-close to states of the form $\ket{a} \otimes \ket{\Phi^+}$ due to~\Cref{thm:characterization-of-internal-test} combined with an application of the triangle inequality.
\end{proof}

\section{Hardness of cloning}

Next, we prove hardness of state generation and witness cloning for the $(\sigma = \mu \otimes \nu, \lambda)$-verification algorithm. 

\subsection{Hardness of state generation}

\begin{theorem}
Assuming $\BQP \not \supseteq \NP$, there is no efficient quantum algorithm which, on input $(\mu, \nu, \lambda)$, produces a state passing the $(\mu \otimes \nu, \lambda)$-verification algorithm if one exists.
\end{theorem}

\begin{proof}
\Cref{cor:passing-all-three-tests} proves that the dimensionality of the accepting subspace for the $(\mu \otimes \nu, \lambda)$-verification algorithm is exactly $m_{\mu \nu \lambda}$, the corresponding Kronecker coefficient\footnote{This gives an alternate proof that exact computation of Kronecker coefficients is contained in $\mathsf{\#BQP}$ proven by Ikemeyer and Subramanian~\cite{ikenmeyer2023remarkquantumcomplexitykronecker}. The original proof by Bravyi~\emph{et. al.}~\cite{kronecker-coeff} proves that $m_{\mu \nu \lambda} d_\lambda$ is a $\mathsf{\#BQP}$ quantity and that $d_\lambda$ is efficiently computable.}. As deciding the positivity of Kronecker coefficients is $\NP$-hard, any efficient algorithm must resolve a $\NP$-hard problem, and therefore prove that $\BQP \supseteq \NP$.
\end{proof}

Furthermore, it should be noted that deciding if $m_{\mu \nu \lambda}$ equals 0 or 1, promised that it is either 0 or 1, is $\mathsf{UNIQUE-NP}$-hard by standard reductions. Furthermore, the previous theorem can be lightly manipulated to show that any efficient algorithm for generating states passing the $(\mu \otimes \nu, \lambda)$-test when $m_{\mu \nu \lambda} = 1$ also implies $\BQP \supseteq \NP$. The proof follows from appealing to the Valiant-Vazirani theorem~\cite{VALIANT198685} which proves that $\NP$ is randomized polynomial-time reducible to $\mathsf{UNIQUE-NP}$:

\begin{corollary}
Assuming $\BQP \not \supseteq \NP$, there is no efficient quantum algorithm which, on input $(\mu, \nu, \lambda)$ such that $m_{\mu \nu \lambda} = 1$, produces a state passing the $(\mu \otimes \nu, \lambda)$-verification algorithm.
\end{corollary}

\subsection{Hardness of cloning}

Morally, in the language of Broadbent, Karvonen, and Lord~\cite{broadbent2023uncloneablequantumadvice} we can think of this result as proof that the ``$0 \rightarrow 1$'' ground-state transformation state complexity class is at least $\NP$-hard. What we would like to state is that the ``$1 \rightarrow 2$'' ground-state transformation is hard as well.

We will show that, assuming a conjecture about cloning maximally entangled states over hidden subspaces, that any efficient cloning algorithm the solutions to the $(\mu \otimes \nu, \lambda)$-verification algorithm would imply that $\BQP \supseteq \NP$. First we observe that the solutions to the $(\mu \otimes \nu, \lambda)$-verification algorithm are a special class of states we call ``hidden maximally entangled states''.

\begin{definition}
For $d_1 \leq d_2$, let $\Pi \in \CC^{d_2 \times d_2}$ be a subspace of dimension $d_1$ and let $\ket*{b_1}, \ldots, \ket*{b_{d_1}}$ be an orthonormal basis for $\Pi$. Then the $\Pi$-maximally entangled state is the state in $\CC^{d_2} \otimes \CC^{d_2}$ defined by
\begin{equation}
    \ket*{\Phi_\Pi} \defeq \frac{1}{\sqrt{d_1}} \sum_{i = 1}^{d_1} \ket*{b_i} \otimes \ket*{b_i^*}.
    \label{eq:U-partial-isometry-hidden-EPR}
\end{equation}
Note that this state is invariant to the choice of basis. 
\end{definition}
When $\Pi$ is a computationally-hard projector such as the ones we previously constructed, we call the state $\ket*{\Phi_\Pi}$ a \emph{hidden maximally entangled state}.
Using $\Xi_\lambda$ as defined in~\eqref{eq:wfs-eq} within~\Cref{fact:weak-fourier-sampling}, the state captured by the $\lambda$-characterization test (\Cref{cor:passing-all-three-tests}) is $\ket{a} \otimes \ket*{\Phi_{\Xi_\lambda}}$ where $\ket{a} \in \CC^{m_{\mu \nu \lambda}}$. In particular, when $m_{\mu \nu \lambda} = 1$, this state passing the test is $\ket*{\Phi_{\Xi_\lambda}}$ up to a global phase.

\begin{conjecture}[Cloning maximally entangled states over hidden subspaces]
    Consider a hidden maximally entangled state $\ket{\Phi_{\Pi}}$ which is the unique state accepted by a completeness $c = 1$ and soundness $s \leq 8/9$ verification circuit~$V$~(\Cref{def:verification-ckt}). Assume there exists a uniform polynomial-time algorithm which takes as input the classical description $\langle V \rangle$ and applies the transformation $\ket{\Phi_{\Pi}}\ket{0} \mapsto \ket{\Phi_{\Pi}}^{\otimes 2}$. Then, we conjecture that there exists an efficient polynomial-time algorithm which takes as input the classical description $\langle V \rangle$ and produces a state in the support of $\Pi$.
    \label{conj:hidden-epr}
\end{conjecture}

One way to view~\Cref{conj:hidden-epr} is that it posits that any white-box circuit for the cloner of $\ket*{\Phi_{\Pi}}$ must leak the description of a circuit which creates a state in the support of $\Pi$. The direct application of~\Cref{conj:hidden-epr} is the following hardness-of-cloning result:
\begin{theorem}
    Assuming~\Cref{conj:hidden-epr} and $\BQP \not \supseteq \NP$, there does not exist an efficient algorithm for cloning the states passing the $(\mu \otimes \nu, \lambda)$-verification algorithm when $m_{\mu \nu \lambda} = 1$.
\end{theorem}

\begin{proof}
Assume there exists an algorithm for cloning the states passing the $(\mu \otimes \nu, \lambda)$-verification algorithm when $m_{\mu \nu \lambda} = 1$. By~\Cref{cor:passing-all-three-tests}, we know that said verification algorithm is a verification circuit (\Cref{def:verification-ckt}) for completeness $c = 1$ and soundness $s \leq 8/9$ with a unique state passing the verification circuit. Furthermore, by the previous argument, said state is a hidden maximally entangled state for hidden subspace $\Pi$~\eqref{eq:U-partial-isometry-hidden-EPR}. Therefore,~\Cref{conj:hidden-epr} directly implies that there exists an efficient circuit for constructing a state in the image of $\Xi_\lambda = \Xi_\lambda^{(\mu \nu)}$. So, running this algorithm produces a quantum certificate that $\dim \Xi_\lambda = m_{\mu \nu \lambda} d_\lambda \neq 0$. As the certificate is efficiently verifiable due to Weak Fourier Sampling (\Cref{fact:weak-fourier-sampling}), we have a $\BQP$ algorithm for solving $\mathsf{UNIQUE-NP}$. Application of the randomized reduction of Valiant and Vazirani~\cite{VALIANT198685} yields $\BQP \supseteq \NP$.
\end{proof}

\section{Discussion}
\subsection{The plausibility of the cloning hidden maximally entangled states conjecture}

Our attempt at proving the hardness of witness cloning from complexity-theoretic assumptions harks on~\Cref{conj:hidden-epr}. For this reason, we believe that this is an ``unfinished work'' and we are hopeful that it will spur interest in the community towards this problem.

\paragraph{The difficulty in cloning a hidden maximally entangled state state} 
We made many attempts at constructing a circuit for cloning a hidden maximally entangled state $\ket*{\Phi_\Pi}$ that did not immediately reveal a circuit for generating a state within $\Pi$ (in an attempt to prove~\Cref{conj:hidden-epr} false). But no attempt was successful. Roughly speaking, if the first copy of the state $\ket*{\Phi_\Pi}$ is supported on registers $A_1A_2$ and we intend to clone the state into $B_1B_2$, there is no transformation we can run on $A_1B_1$ and on $A_2B_2$ to clone the state as the entanglement between $A_1B_1$ and $A_2B_2$ will persist. Morally, it appears that cloning the state requires disentangling the original state across $A_1A_2$. However, we believe that such a disentangler run backwards would generate a state in $\Pi$. It follows that a cloning circuit that does not reveal a generating circuit must efficiently rotate the standard basis vectors into an orthonormal set of vectors that includes $\ket*{\Phi_\Pi}$ without constructing a state for generating states in $\Pi$.
However, we emphasize that these statements are only intuitions and we were unable to make any of them rigorous.

\paragraph{A white-box proof is necessary for~\Cref{conj:hidden-epr}} It is important to note that ~\Cref{conj:hidden-epr} cannot follow a black-box reduction. The ideal resolution of~\Cref{conj:hidden-epr} would be in the \emph{black-box} setting where one would show that a black-box (or oracle) gate for cloning the hidden maximally entangled state can be used $O(\poly(n))$ times as a subroutine in an algorithm constructing a state in $\Pi$. However, we know that this is not true due to the observations of Nehoran and Zhandry~\cite{nehoran_et_al:LIPIcs.ITCS.2024.82}. Nehoran and Zhandry show the existence of a quantum oracle model in which cloning of verifiable states is easy and yet construction of said verifiable states is hard. Therefore, any proof of~\Cref{conj:hidden-epr} must be quantum non-relativizing. 
At a high level, the Nehoran and Zhandry~\cite{nehoran_et_al:LIPIcs.ITCS.2024.82} intuition when applied to~\Cref{conj:hidden-epr} is to consider a quantum oracle which performs the 2-dimensional reflection which reflects $\ket{\Phi_\Pi}^{\otimes 2}$ onto $\ket{\Phi_\Pi}\ket{0}$. The action of the oracle in the space orthogonal to this 2-dimensional subspace is identity. Equivalently, with access to this quantum oracle cloning is easy and yet this oracle cannot be used to create a state within $\Pi$. 

\subsection{Basing the argument on average-case vs worst-case assumptions}
In our research of this problem, we came across the observation that it is appears significantly easier to prove hardness-of-cloning statements from average-case assumptions rather than worst-case assumptions. Ideally, one would like to show that having access to an oracle for witness cloning would drastically improve the power of $\BQP$ (to include say $\NP$ or $\QMA$). However, it isn't clear how to use the oracle effectively as the only states a polynomial-time algorithm could query the oracle are already are efficiently constructable; therefore. In essence, this becomes a chicken-and-egg problem: a cloner is only effective for generating hard states given a hard state. 

Instead, armed with an average-case assumption, we can create a random ensemble over states which are each individually verifiable and yet hard to construct on average. Crucially, the ensemble must be superpolynomial in size. This is the basis of quantum lightning schemes. In lightning schemes, the average-case hardness assumption is used to argue that an efficient cloner for most states in the ensemble cannot exist.
Notice that a worst-case assumption would not suffice as it may only apply to a negligible fraction of the states in the ensemble. Since we do not know of average-case hardness statements for any $\NP$-complete problems, proving the hardness-of-cloning from an assumption like $\BQP \not \supseteq \NP$ appears to require fundamentally different techniques. 
One of the appealing properties of Kronecker coefficients which inspired our investigation is that the average-case complexity of computing Kronecker coefficients is unknown.

If we were able to prove an average-case statement for computing Kronecker coefficients, then Kronecker coefficient may lead to a new quantum lightning construction. This is because measuring the maximally entangled state $\ket{\Phi^+} \in \CC^{d_\mu d_\nu \times d_\mu d_\nu}$ according to the weak Fourier sampling POVM~(\Cref{fact:weak-fourier-sampling}) for representation $\rho^\mu \otimes \rho^\nu \otimes \II$ generates a state $\ket*{\Psi_\lambda^{(\Phi^+)}}$ (defined in~\eqref{eq:passing-improved-construction}) for a random $\lambda$. The probability of measuring irrep $\lambda$ can be calculated as $\frac{d_\lambda}{d_\mu d_\nu} m_{\mu \nu \lambda}$.

\section*{Acknowledgements}
We want to thank Sergey Bravyi for many fruitful discussions.

\bibliography{references}
\bibliographystyle{myhalpha}

\appendix

\section{Quick Reference}
\subsection{Representation Theory}

Most of the representation theory results used in this work stem from Schur's lemma:
\begin{fact}[Schur's lemma]
For any $1 \leq i_1, j_1 \leq d_{\lambda_1}$ and $1 \leq i_2, j_2, \leq d_{\lambda_2}$ for irreps $\lambda_1, \lambda_2$ of $G$,
\begin{equation}
    \sum_{g \in G} \rho^{\lambda_1}_{i_1j_1}(g)^* \rho^{\lambda_2}_{i_2 j_2}(g) = \frac{\abs{G}}{d_{\lambda_1}} \cdot \delta_{\lambda_1 \lambda_2} \delta_{i_1 i_2} \delta_{j_1 j_2}.
\end{equation}
\label{fact:schur}
\end{fact}
A useful consequence of Schur's lemma is the following corollary about the orthogonality of characters. It follows by taking the sum over all $i_1 = j_1$ and $i_2 = j_2$. Then,
\begin{xalign}[eq:orthog-of-chars]
    \sum_{g \in G} \chi^{\lambda_1}(g)^* \chi^{\lambda_2}(g) &= \sum_{i=1}^{d_{\lambda_1}} \sum_{j = 1}^{d_{\lambda_2}} \sum_{g \in G} \rho^{\lambda_1}_{ii}(g)^* \rho^{\lambda_2}_{jj}(g) \\
    &= \sum_{i=1}^{d_{\lambda_1}} \sum_{j = 1}^{d_{\lambda_2}} \frac{\abs{G}}{d_{\lambda_1}} \cdot \delta_{\lambda_1 \lambda_2} \delta_{ij}. \\
    &= \abs{G} \cdot \delta_{\lambda_1 \lambda_2}.
\end{xalign}
We can also derive the following result:
\begin{xalign}[eq:ip-of-char-and-rep-v2]
    \sum_{g \in G} \chi^{\lambda_1}(g)^* \rho^{\lambda_2}(gh) &= \sum_{i = 1}^{d_{\lambda_1}} \sum_{j,k, \ell=1}^{d_{\lambda_2}} \sum_{g \in G} \rho_{ii}^{\lambda_1} (g)^* \rho_{jk}^{\lambda_2}(g) \rho
    _{k\ell}^{\lambda_2}(h) \ketbra{j}{\ell} \\
    &= \sum_{i,\ell} \frac{\abs{G}}{d_{\lambda_1}} \delta_{\lambda_1 \lambda_2} \rho_{i\ell}^{\lambda_2}(h) \ketbra{i}{\ell} \\
    &= \frac{\abs{G}}{d_{\lambda_1}} \delta_{\lambda_1 \lambda_2} \rho^{\lambda_2}(h).
\end{xalign}

\subsection{Vectorization calculations}
\label{subsec:vectorization-math}

For any $d \times d$ matrix $A$, let us define $\ket{\mathrm{vec} \ A} \defeq \sum_{i,j = 1}^d A_{ij} \ket{i}\ket{j}$. Then $\norm{\ket{\mathrm{vec} \ A}} = \norm{A}_F$,
the Frobenius norm of the matrix.
Additionally, the inner product between two vectorizations is equal to
\begin{equation}
    \bra{\mathrm{vec} \ A} \ket{\mathrm{vec} \ B} = \sum_{i,j = 1}^d A_{ij}^* B_{ij} = \sum_{ij}^d A^\dagger_{ji} B_{ij} = \tr(A^\dagger B) = \ev{A, B}_F.
\end{equation}
And for $B, C \in \CC^{d \times d}$, $B \otimes C \ket{\mathrm{vec}\ A} = \ket*{\mathrm{vec}\ BAC^\top}$.
\section{Omitted Proofs}
\label{sec:omitted-pfs}

\factwfs*

\begin{proof}[Proof of \Cref{fact:weak-fourier-sampling}]
By direct calculation, we can verify the orthogonality of the projectors in the POVM:
\begin{xalign}
    \Xi_\lambda \Xi_\mu &= \frac{d_\lambda d_\mu}{|G|^2} \sum_{g \in G} \sum_{h \in G} \chi^{\lambda}(g)^*\chi^{\mu}(h)^*  \sigma(gh) \\ &= \frac{d_\lambda d_\mu}{|G|^2} \sum_{g \in G} \sum_{\ell \in G} \chi^{\lambda}(g)^*\chi^{\mu}(g^{-1}\ell)^*  \sigma(\ell) \\
    &=  \frac{d_\lambda d_\mu}{|G|^2}  \sum_{\ell \in G} \sum_{i \in [d_\lambda]} \sum_{j,k \in [d_\lambda]} \left( \sum_{g \in G} \rho^{\lambda}(g)^*_{ii} \rho^{\mu}(g)_{kj} \right) \rho^{\mu}(\ell)^*_{kj}  \sigma(\ell) \\
    &= \delta_{\mu \lambda} \frac{d_\lambda}{|G|} \sum_{\ell \in G} \sum_{i \in [d_\lambda]} \sum_{j,k \in [d_\lambda]} \delta_{ik} \delta_{ij} \rho^{\lambda}(\ell)^*_{kj} \sigma(\ell) \\
    &= \delta_{\mu \lambda}  \frac{d_\lambda}{|G|} \sum_{\ell \in G}  \chi^{\lambda}(\ell)^* \sigma(\ell) \\
    &= \delta_{\mu \lambda}  \Xi_\lambda,
\end{xalign}
where the third and fourth rows follow from the definition of $\chi^\lambda$ and Schur's lemma~(\Cref{fact:schur}). This shows that $\lbrace \Xi_\lambda \rbrace_{\lambda \in \Ii_G }$ is a set of orthogonal projectors. Since $d_\lambda = \chi^{\lambda}(e)$,
\begin{xalign}
    \tr\qty(\sum_{ \lambda \in \Ii_G} \Xi_\lambda) &= \tr\qty( \sum_{\lambda \in \Ii_G} \sum_{g \in G} \frac{d_\lambda}{|G|} \chi^\lambda(g)^* \sigma(g)) \\
    &= \sum_{\lambda \in \Ii_G} \frac{d_\lambda}{\abs{G}} \cdot \frac{\abs{G}}{d_\lambda} \tr(\rho^\lambda(e)) \label{eq:simplification-of-char}\\
    &= \sum_{\lambda \in \Ii_G} m_{\sigma \lambda} d_\lambda = D. \label{eq:sum-of-chars}
\end{xalign}
We applied~\eqref{eq:ip-of-char-and-rep-v2} to achieve~\eqref{eq:simplification-of-char} and~\eqref{eq:sum-of-chars} follows from the orthogonal decomposition. 
By orthogonality and the trace equation, it follows that $\sum_{\lambda \in \Ii_G} \Xi_\lambda = \II_D$, and therefore this is a POVM over $D$ dimensions.
This measurement can be implemented efficiently whenever the quantum $G$-Fourier transform can be implemented efficiently. We use Harrow's generalized phase estimation algorithm\cite{harrow2005applicationscoherentclassicalcommunication}. This circuit uses two registers, a control and a target. The control register is initialized as thee uniform superposition over $G$, and the controlled action of $\rho$ is applied to the target register. The irrep label is then measured in the Fourier basis to conclude the computation: 

$$
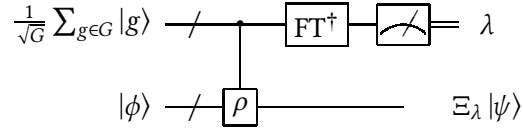
\qquad \qquad \Qcircuit @C=1em @R=1.5em {
    & \lstick{\frac{1}{\sqrt{G}} \sum_{g \in G} \ket{g}} & {/} \qw & \ctrl{1} & \gate{\mathrm{FT}^\dagger} & \meter & \cw & \lambda \\
    & \lstick{\ket{\phi}} & {/} \qw & \gate{\rho} & \qw & \qw &  & \Xi_\lambda \ket{\psi}
}
$$

\noindent Similarly to~\cite[Theorem 1]{larocca2024quantumalgorithmsrepresentationtheoreticmultiplicities}, we can characterize the action of this circuit using a superoperator $E_\lambda^\dag (\cdot) E_\lambda$ with elements $E_\lambda$:
    \begin{align}
         E_\lambda &= \frac{1}{\sqrt{|G|}} \sum_{g \in G} \left( \Pi_\lambda \mathrm{FT}\ket{g} \right)_C \otimes \rho(g)_T,
    \end{align}
    where $\Pi_\lambda$ measures the irrep label in the Fourier basis:
    \begin{align}
        \Pi_\lambda := \ketbra{\lambda} \otimes \II_{d_\lambda^2} = \sum_{i,j \in [d_\lambda]} \ket{\lambda,i,j} \bra{\lambda, i,j} 
    \end{align}
    and $\mathrm{FT}$ is the Fourier transform for $G$ previously defined in~\eqref{eq:qft}. Therefore, 
 \begin{align}
     \Pi_\lambda \mathrm{FT} \ket{g} = \sum_{i,j \in [d_\lambda]} \sqrt{\frac{d_\lambda}{|G|}} \rho^\lambda_{ij}(g) ^*\ket{\lambda, i, j}
 \end{align} 
 We now show that: $E_\lambda^\dag E_\lambda = \Xi_\lambda$. 
 \begin{xalign}
     E_\lambda^\dag E_\lambda &= \frac{1}{|G|} \sum_{g,h \in G} \bra{h}{\mathrm{FT}^\dag \Pi_\lambda \mathrm{FT}}\ket{g} \rho(gh^{-1}) \\ 
     &= \frac{d_\lambda}{|G|^2} \sum_{g,h \in G} \sum_{i, j \in [d_\lambda]} \rho^{\lambda}_{ij}(h) \rho^\lambda_{ij}(g)^* \rho(gh^{-1}) \\
     &= \frac{d_\lambda}{|G|^2} \sum_{g, h \in G} \chi^\lambda(gh^{-1})^*\rho(gh^{-1}) = \frac{d_\lambda}{|G|} \sum_{g \in G} \chi^\lambda(g)^* \rho(g).
 \end{xalign}
 In particular, trace cyclicity then implies that that for any state $\sigma$: 
 \begin{align}
     \tr(E_\lambda \sigma E_\lambda^\dag) = \tr(\Xi_\lambda \sigma).
 \end{align}
 
\end{proof}

\factefficientprojector*

\begin{proof}[Proof of \Cref{fact:efficientprojector}]
We argue that we can efficiently implement the action of $\rho^\mu \otimes \rho^\nu$ in a controlled way and combine this with an efficient $S_n$ Fourier transform in \Cref{fact:weak-fourier-sampling}. We use the Young-Yamanouchi basis for $\rho^\mu$. Rows and columns are labeled by distinct semi-standard Young tableaux (SSYTs) of shape $\mu$ and the representation matrix $\rho^\mu(\sigma_i)$ for a transposition $\sigma_i$ was defined by James~\cite{james1984}\footnote{Additional exposition can be found in~\cite{jordan2009fastquantumalgorithmsapproximating}.} as :
   \begin{align}
       \rho^\mu(\sigma_i) \ket{k} &=  \frac{1}{\tau_i} \ket{k} + \delta[\ell \text{ is SSYT}]\sqrt{1-\frac{1}{\tau_i^2}} \ket{\ell},
       \label{Eq:Young-Yamanouchi}
   \end{align}
   where $\ell$ is a Young tableau obtained from $k$ by swapping $i$ and $i+1$;
   \begin{align}
       \delta[\ell \text{ is SSYT}] &= \begin{cases} 1 \text{ if } \ell \text{ is SSYT} \\ 0 \text{ otherwise } \end{cases}
   \end{align}
    and $\tau_i$ is the \emph{axial distance} of $i$ and $i+1$ in $k$: the number of steps between $i$ and $i+1$ counting each step up or to the right as $-1$ and each step to left or down as $+1$. Both $\tau_i$ and $\delta[\ell \in \text{SSYT}]$ can be computed in $\poly(n)$ time by a classical algorithm and implementing this reversibly gives a $\poly(n)$ quantum circuit for $\rho^\mu(\sigma_i)$. Additionally, we can efficiently implement the controlled version of this circuit. Since any permutation is a product of $O(n^2)$ nearest-neighbor transpositions, composition of the transposition circuits gives a poly(n) quantum circuit that implements the action of the irreducible matrix for an arbitrary permutation. Finally, $\sigma = \rho^\mu \otimes \rho^\nu$ can be implemented by two copies of such circuit that uses the same control register. The quantum $S_n$-Fourier transform can be implemented in polynomial time~\cite{beals97}. Combined with \Cref{fact:weak-fourier-sampling}, it follows that the measurement $\lbrace \Xi_\lambda \rbrace_{\lambda \vdash n}$ can be implemented efficiently for $\sigma = \rho^{\mu} \otimes \rho^{\nu}$.
\end{proof}

\lemitfiip*

\begin{proof}[Proof of~\Cref{lem:interior-test-for-irrep-identity-prod}]
Define the matrix $X$ of Frobenius norm 1 such that $\ket{\psi} = \ket{\mathrm{vec} \ X}$.
The state being tested is $\ket{\tau} = \frac{1}{\sqrt{\abs{G}}} \sum_{k \in G} \ket{k} \otimes \ket{\mathrm{vec} \ X}$. Then, by direct calculation,
\begin{xalign}
   \ev{U}{\tau} &= \frac{1}{\abs{G}} \sum_{k} \ev{\sigma(k) \otimes \sigma(k)^*}{\mathrm{vec} \ X} \\
   &= \tr \qty(X^\dagger \qty(\frac{1}{\abs{G}} \sum_k \sigma(k) X \sigma(\inv{k}))) \\
   &= \ev{X, \mathcal{E}(X)}_F
\end{xalign}
where we have defined
\begin{equation}
    \mathcal{E}(X) = \mathcal{E}_{\sigma}(X) \defeq \frac{1}{\abs{G}} \sum_{k \in G} \sigma(k) X \sigma(\inv{k}),
\end{equation}
The probability of the internal state testing protocol accepting the 1-bit phase estimation is
\begin{align}
\Pr\qty[\textrm{internal state test accepts } \ket*{\Psi_\mathrm{in}^{(\psi)}}] = \frac{1}{2} + \frac{1}{2} \abs{\ev{X, \mathcal{E}(X)}_F}^2 \label{eq:prob-internal-state-test-accepts}.
\end{align}
Now, we apply the decomposition of $\sigma = \II_m \otimes \sigma_1$. So, we express the matrix $X$ as
\begin{equation}
X = \sum_{i, j = 1}^m \ketbra{i}{j} \otimes a_{ij} X^{(ij)}
\end{equation}
such that $\tr(X^{(ij)}) \in \RR^+$ and $\norm{X^{(ij)}}_F = 1$. Since, we only consider $X$ such that $\norm{X}_F = 1$, This is equivalent to $\norm{A}_F = 1$ where $A$ is the matrix formed from coefficients $(a_{ij})$. Therefore,
\begin{xalign}
X^\dagger &= \sum_{ij} \ketbra{i}{j} \otimes  a_{ji}^\dagger {X^{(ji)}}^\dagger, \\
\text{and } \mathcal{E}_\sigma(X) &= \Exp_k \qty(\II_m \otimes \sigma_1(k)) X \qty(\II_m \otimes \sigma_1(\inv{k})) \\ &= \sum_{ij} \ketbra{i}{j} \otimes a_{ij} \mathcal{E}_{\sigma_1}(X^{(ij)}).
\end{xalign}
Therefore,
\begin{xalign}
X^\dagger \mathcal{E}_\sigma(X) &= \qty(\sum_{i_1 j_1} \ketbra{i_1}{j_1} \otimes a_{j_1 i_1}^\dagger {X^{(j_1 i_1)}}^\dagger) \qty(\sum_{i_2j_2} \ketbra{i_2}{j_2} \otimes a_{i_2 j_2} \mathcal{E}_{\sigma_1}(X^{i_2 j_2})) \\
&= \sum_{i_1, j_2} \ketbra{i_1}{j_2} \otimes \qty(\sum_k a_{ki_1}^\dagger a_{kj_2} {X^{(k i_1)}}^\dagger \mathcal{E}_{\sigma_1}(X^{(kj)})), \\
\text{and } \tr\qty(X^\dagger \mathcal{E}_\sigma(X)) &= \sum_{ij} \abs{a_{ij}}^2 {X^{(ij)}}^\dagger \mathcal{E}_{\sigma_1}(X^{(ij)}).
\end{xalign}
Therefore, the quantity of interest equals
\begin{xalign}
\ev{X, \mathcal{E}(X)}_F &= \tr \qty(X^\dagger \mathcal{E}(X)) \\
&= \sum_{ij} \abs{a_{ij}}^2 \ev{X^{(ij)}, \mathcal{E}_{\sigma_1}(X^{(ij)})}_F \\
&= \frac{1}{D_1} \sum_{ij} \abs{a_{ij}}^2 \tr(X^{(ij)})^2. \label{eq:trivial-commutator-consequence}
\end{xalign}
where \eqref{eq:trivial-commutator-consequence} follows since $\sigma_1$ has a trivial commutator and therefore
\begin{equation}
\mathcal{E}_{\sigma_1}(X^{(ij)}) = \frac{\tr(X^{(ij)})}{D_1} \II_{D_1}.
\end{equation}
We now will show that if $\abs{ \ev{X, \mathcal{E}(X)}_F} \approx 1$, then the state $\ket*{\mathrm{vec} \ X}$ is $\approx \ket*{\mathrm{vec} \ A} \otimes \ket{\Phi^+}$. 
Since, we assumed $\tr(X_{ij})$ was $\in \RR^+$, its closest unit vector proportional to $\II$ is $\II/\sqrt{D_1}$.
Then, as we are considering the Frobenius norm, it follows that
\begin{xalign}
    \norm{X^{(ij)} - \frac{1}{\sqrt{D_1}} \II_{D_1}}_F &\leq \arccos \qty( \ev{\frac{1}{\sqrt{D_1}} \II_{D_1}, X^{(ij)}}_F) \\
    &\leq 2 \sin \arccos(\ev{\frac{1}{\sqrt{D_1}} \II_{D_1}, X^{(ij)}}_F) \label{eq:sin-rule} \\
    &= 2 \sqrt{1 - \frac{\tr(X^{(ij)})^2}{D_1}}.
\end{xalign}
We use the fact that $x \leq 2 \sin x$ for small angles in \eqref{eq:sin-rule}. Therefore, the Frobenius distance between the $X$ and a product matrix can be calculated as
\begin{xalign}
\norm{X - A \otimes \frac{1}{\sqrt{D_1}} \II_{D_1}}_F &= \sqrt{ \sum_{\textbf{}ij} \norm{a_{ij} X^{(ij)} - a_{ij} \frac{1}{\sqrt{D_1}} \II_{D_1} }_F^2} \\
&= \sqrt{ \sum_{ij} \abs{a_{ij}}^2 \norm{X^{(ij)}- \frac{1}{\sqrt{D_1}} \II_{D_1}}_F^2} \\
&= 2 \sqrt{\sum_{ij} \abs{a_{ij}}^2 \qty(1 - \frac{\tr(X^{(ij)})^2}{D_1})} \\
&= 2 \sqrt{1 - \ev{X, \mathcal{E}(X)}_F} \\
&\leq 2 \sqrt{1 - \ev{X, \mathcal{E}(X)}_F^2} \label{eq:dist-from-product}
\end{xalign}
This will equal the distance between the states $\ket*{\mathrm{vec} \ X}$ and $\ket*{\mathrm{vec} \ A} \otimes \ket{\Phi^+}$. From \eqref{eq:prob-internal-state-test-accepts}, we see that if we pass the test with probability $1 - \eps$, then $1 - \ev{X, \mathcal{E}(X)}_F^2 \leq 2\eps$. Therefore,
\begin{equation}
    \norm{X - A \otimes \frac{1}{\sqrt{D_1}} \II_{D_1}}_F \leq  2 \sqrt{2\eps}.
\end{equation}
Converting the matrix equation over Frobenius norm to an equation over states and $\ell_2$-norm, completes the proof.
\end{proof}

\end{document}